\documentclass[conference]{IEEEtran}[10pt]

\usepackage[dvips]{color}
\usepackage{epsf}
\usepackage{times}
\usepackage{epsfig}
\usepackage{graphicx}
\usepackage{amsmath}
\usepackage{amssymb}
\usepackage{amsxtra}
\usepackage{here}
\usepackage{rawfonts}
\usepackage{times}
\usepackage{url}

\newtheorem{theorem}{\bf Theorem}

\newtheorem{definition}{\bf Definition}

\IEEEoverridecommandlockouts
\begin{document}

\title{ A Distributed Merge and Split Algorithm for\\ Fair Cooperation in Wireless Networks}

\author{\authorblockN{ Walid Saad$^1$, Zhu Han$^2$, M{\'e}rouane Debbah$^3$, and Are Hj{\o}rungnes$^1$ }
\authorblockA{\small
$^1$UniK - University Graduate Center, University of Oslo, Kjeller, Norway, Email: \url{{saad, arehj}@unik.no}\\
$^2$ Electrical and Computer Engineering Department, Boise State
University, Boise, USA, Email:\url{zhuhan@boisestate.edu}\\
$^3$ Alcatel-Lucent Chair in Flexible Radio,  SUP{\'E}LEC, Gif-sur-Yvette, France, Email: \url{merouane.debbah@supelec.fr}\vspace{-0.75cm}
 }%
   \thanks{This work was supported by the Research Council of Norway
    through the project 183311/S10 entitled "Mobile-to-Mobile Communication Systems (M2M)", the AURORA project 18778/V11 entitled "Communication under uncertain topologies" and the OptiMo project 176773/S10 entitled "Optimized Heterogeneous Multiuser MIMO Networks".}}

\date{}

\maketitle

\begin{abstract}
This paper introduces a novel concept from coalitional game theory which allows the dynamic formation of coalitions among wireless nodes. A simple and distributed merge and split algorithm for coalition formation is constructed. This algorithm is applied to study the gains resulting from the cooperation among single antenna transmitters for virtual MIMO formation. The aim is to find an ultimate transmitters coalition structure that allows cooperating users to maximize their utilities while accounting for the cost of coalition formation. Through this novel game theoretical framework, the wireless network transmitters are able to self-organize and form a structured network composed of disjoint stable coalitions. Simulation results show that the proposed algorithm can improve the average individual user utility by 26.4\% as well as cope with the mobility of the distributed users.
\end{abstract}

\section{Introduction}
Recently, cooperation between mobile devices has been one of the main
activities of research work which tackled different aspects of cooperation at
different layers. For instance, the problem of cooperation among single antenna
receivers for virtual MIMO formation has been studied in
\cite{MA02} using coalitional game theory. The authors in
\cite{MA02} proved that for the receiver coalition game in a
Gaussian interference channel and synchronous CDMA multiple access
channel (MAC), a stable grand coalition of all users can
form if \emph{no cost} for cooperation is considered. Cooperation among single antenna transmitters and receivers in ad
hoc networks has also been studied in \cite{GC01}. The authors inspected the capacity gains that users can achieve
while cooperating; namely at the transmitter side. Cooperation in wireless networks was also exploited at higher
layers such as the network layer using game theory. For instance, cooperation in routing protocols was tackled in~\cite{RO00} where selfish nodes are allowed to cooperate for reducing the routing energy cost. In \cite{RO01} and \cite{Ibrahim},
the nodes along the route can cooperate with each other in order
to improve the link quality or reduce the power consumption.
Cooperation for packet forwarding is studied in \cite{BC00} and
\cite{BC01} using cooperative game theory, repeated game
theory, and machine learning. Auction theory is used in \cite{REL00} for
optimizing relay resource allocation in cooperative networks.
Thus, previous work mainly focused on the study of the gains resulting from virtual MIMO as well as from higher layer cooperation. For virtual MIMO,
the usage of coalitional games was limited to the study of the formation of the grand coalition when the users cooperate without cost. However, due to cooperation costs, it might not be beneficial for users located far away from each other to cooperate.

The main contribution of this article is to derive a fair cooperation strategy among distributed single antenna transmitters, which will allow these users to self-organize into structured coalitions while maximizing their utilities with cost. For this purpose, we construct a coalition formation algorithm
based on a novel concept from coalitional game theory which, to our knowledge, have not been used in the wireless world yet. A simple and distributed merge and split algorithm is devised for forming the coalitions. Through this algorithm, we seek to find the possible coalition structures in the transmitter cooperation game.  The convergence of the algorithm is proven and the stability of the resulting partitions is investigated. Different fairness criteria for dividing the extra benefits among the coalition users are also discussed.

The rest of this paper is organized as follows: Section~\ref{sec:prob} presents
the system model. In Section
\ref{sec:form}, we present the proposed algorithm, prove its properties and discuss the
fairness criteria. Simulation results are presented and analyzed
in Section \ref{sec:sim}. Finally, conclusions are drawn in
Section \ref{sec:conc}.

\section{System Model} \label{sec:prob}
Consider a network having $M_{t}$ single antenna transmitters, e.g.~mobile users, sending data in the uplink to a fixed receiver, e.g.~a base station, with $M_{r}$ receivers (Multiple Access Channel). Denote $N=\{1\ldots M_{t}\}$ as the set of all $M_{t}$ users in the network, and let  $S \subset N$ be a coalition consisting of
$|S|$ users. We consider a TDMA transmission in the network, thus,  in a non-cooperative manner, the $M_t$ users require a time scale of $M_t$ slots since every user occupies one slot. When cooperating, the single antenna transmitters form different disjoint coalitions and the coalitions will subsequently transmit in a TDMA manner, that is one coalition per transmission. Thus, during the time scale $M_t$, each coalition is able to transmit within all the time slots previously held by its users.
In addition, similar to \cite{GC01} we define a fixed transmitting power constraint $\tilde{P}$ per time slot (i.e.~per coalition) which is the power available for all transmitters that will occupy this slot. If a coalition (viewed as a single user MIMO after cooperation) will occupy the slot, part of the transmit power constraint will be used for actual transmission while the other part will constitute a cost for the exchange of information between the members of the coalition. This cost is taken as the sum of the powers required by \emph{each} user in a coalition $S$ to broadcast to its corresponding farthest user inside $S$. For instance, the power needed for broadcast transmission between a user $i \in S$ and its corresponding farthest user $\hat{i} \in S$ is given by
\begin{equation}\label{eq:scost}
\bar{P}_{i,\hat{i}} = \frac{\nu_0 \cdot \sigma^{2}}{h^{2}_{i,\hat{i}}} ,
\end{equation}
where $\nu_0$ is a target SNR for information exchange, $\sigma^{2}$ is the noise variance and $h_{i,\hat{i}}=\sqrt{\kappa/d_{i,\hat{i}}^{\alpha}}$ is the path loss between users $i$ and $\hat{i}$; $\kappa$ being the path loss constant, $\alpha$ the path loss exponent and $d_{i,\hat{i}}$ the distance between users $i$ and $\hat{i}$. In consequence, the total power cost for a coalition $S$ having $|S|$ users is given by $\hat{P}_{S}$ as follows
\begin{equation}\label{eq:cost}
\hat{P}_{S} = \sum_{i=1}^{|S|}\bar{P}_{i,\hat{i}}.
\end{equation}
It is interesting to note that the defined cost depends on the location of the users and the size of the coalition; hence a higher power cost is incurred whenever the distance between the users increases or the coalition size increases. Thus, the actual power constraint per coalition $S$ is given by
\begin{equation}\label{eq:constr}
P_{S} = (\tilde{P} - \hat{P}_{S})^{+} , \textrm{with } a^{+} \triangleq \max{(a,0)}
\end{equation}

In the considered TDMA system, each coalition transmits in a time slot, hence, perceiving no interference from other coalitions during transmission.
Therefore, in a time slot, the sum-rate of the virtual MIMO system formed by a coalition $S$, assuming Gaussian signalling and under a power constraint $P_{S}$, is given
by \cite{TE00} as
\begin{equation}\label{eq:maxmutu}
C_{S}=\max_{\boldsymbol{Q}_{S}}I(\boldsymbol{x}_{S};\boldsymbol{y}_{S})  = \max_{\boldsymbol{Q}_{S}}\log{\operatorname{det}(\boldsymbol{I}_{M_r} + \boldsymbol{H}_{S}\cdot\boldsymbol{Q}_{S}\cdot \boldsymbol{H}_{S}^{\dag} )} ,
\end{equation}
where $\boldsymbol{x}_{S}$ and $\boldsymbol{y}_{S}$ are, respectively, the transmitted
and received signal vectors of coalition $S$, $\boldsymbol{Q}_{S} = \operatorname{E}{[\boldsymbol{x}_{\cal{S}}\cdot \boldsymbol{x}^{\dag}_{\cal{S}}]}$ is the covariance of $\boldsymbol{x}_{S}$ with  $\operatorname{tr}[\boldsymbol{Q}_{S}] \le P_{S}$ and $\boldsymbol{H}_{S}$ is the $M_r \times M_t$ channel matrix with $\boldsymbol{H}_{S}^{\dag}$ its conjugate transpose.

In this work, we consider a path-loss based  deterministic
channel matrix $\boldsymbol{H}_{S}$ assumed perfectly known at the
transmitter and receiver with each element $h_{i,j} =
\sqrt{\kappa/d_{i,j}^{\alpha}}$ with $d_{i,j}$ the distance between transmitter $i$ and receiver $j$. For such a channel, the work in \cite{TE00} shows that the maximizing input signal covariance $\boldsymbol{Q}_{S}$ is given by $\boldsymbol{Q}_{S} = \boldsymbol{V}_{S}\boldsymbol{D}_{S}\boldsymbol{V}_{S}^{\dag}; \textrm{ with } \operatorname{tr}[\boldsymbol{D}_{S}]  = \operatorname{tr}[\boldsymbol{Q}_{S}]$ where $\boldsymbol{V}_{S}$ is the unitary matrix given by the singular value decomposition of $\boldsymbol{H}_S=\boldsymbol{U}_S \boldsymbol{\Sigma}_S \boldsymbol{V}_S^{\dag}$. $\boldsymbol{D}_S$ is an $M_t \times M_t$ diagonal matrix given by $\boldsymbol{D}_S = \operatorname{diag}(D_1,\ldots,D_{K},0,\ldots,0)$ where $K\le\min{(M_r,M_t)}$ represents the number of positive singular values of $\boldsymbol{H}_{S}$ (eigenmodes) and each $D_i$ given by
\begin{equation}\label{eq:q}
   D_{i} = (\mu - \lambda_{i}^{-1})^{+}.
   \end{equation}
$\mu$ is determined by water-filling to satisfy the coalition power
constraint $  \operatorname{tr}[\boldsymbol{Q}_{S}] = \operatorname{tr}[\boldsymbol {D}_{S}]   = \sum_{i} D_{i} = P_{S}$  and  $\lambda_i$ represents the $i$th eigenvalue of $\boldsymbol{H}_{S}^{\dag}\boldsymbol{H}_{S}$. Hence, based on \cite{TE00}, the resulting capacity for a coalition $S$ is given by $C_{S} = \sum_{i=1}^{K} (\log{(\mu\lambda_i)})^{+}$.

Consequently, over the TDMA time scale of $M_t$, for every coalition $S\subset N$, we define the utility function as
\begin{equation}\label{eq:utility}
v(S) = \begin{cases} |S| \cdot C_{S}, & \mbox{if } P_{S} > 0,\\ 0, &
\mbox{otherwise}. \end{cases}
\end{equation}
This utility represents the total capacity achieved by coalition $S$ during the time scale $M_t$ while accounting for the cost through the power constraint. The
second case states that if the power cost within a coalition is larger than or equal the constraint, then the
coalition cannot form. Thus, we have a coalitional game ($N$,$v$) with a transferable utility and we seek, through coalition formation, a coalition structure that will allow the users to maximize their utilities in terms of rate with cost in terms of power.

\section{Coalition Formation} \label{sec:form}
\subsection{Coalition Formation Algorithm}
Unlike existing literature, in the proposed transmitter cooperation ($N$,$v$) game, we will prove that the
grand coalition cannot form due to cost.
\begin{definition}
A coalitional game $(N,v)$ with a transferable utility is said to be superadditive if for any two disjoint coalitions $S_{1}, S_{2}
\subset N$, $v(S_{1} \bigcup S_{2}) \ge v(S_{1}) + v(S_{2})$.
\end{definition}
\begin{theorem}The proposed transmitter ($N$,$v$) coalitional game with cost is, in general, non-superadditive.
\end{theorem}
\begin{proof}
Consider two disjoint coalitions $S_1 \subset N$  and $S_2 \subset N$ in the network, with the users of $S_1 \bigcup S_2$ located far enough to yield a power cost per (\ref{eq:cost}) $\hat{P}_{S_1 \bigcup S_1} > \tilde{P}$. Therefore, by (\ref{eq:constr}) $P_{S_1 \bigcup S_2} =0$ yielding $v(S_1 \bigcup S_2) = 0 < v(S_1) + v(S_2)$ (\ref{eq:utility}); hence the game is not superadditive.
\end{proof}
\begin{definition}
A payoff vector $\boldsymbol{z} = (z_{1},\ldots,z_{M_t})$ is said
to be \emph{group rational} or efficient if $\sum_{i=1}^{M_t}z_{i} = v(N)$. A
payoff vector $\boldsymbol{z}$ is said to be \emph{individually rational} if the
player can obtain the benefit no less than acting alone, i.e.
$z_i \ge v(\{i\}) \forall i$. An \emph{imputation} is a payoff vector satisfying
the above two conditions.
\end{definition}
\begin{definition}
An imputation $\boldsymbol{z}$ is said to be unstable through a
coalition $S$ if $v(S)$$>\sum_{i\in S}z_i$, i.e., the players have
incentive to form coalition $S$ and reject the proposed $\boldsymbol{z}$.
The set $\mathcal{C}$ of stable imputations is called the {\em core}, i.e.,
\begin{equation}
\mathcal{C}=\left\{\boldsymbol{z}:\sum_{i \in N }z_i=v(N) \mbox{ and
} \sum_{i\in S}z_i\geq v(S)\ \forall S\subset N\right\}.
\end{equation}
A non-empty core means that the players have an incentive to form the grand coalition.
\end{definition}

\begin{theorem}In general, the core of the proposed ($N$,$v$) coalitional game is empty.
\end{theorem}
\begin{proof}
Similarly to the proof of Theorem 1, consider a TDMA network composed of only two disjoint coalitions $S_1$ and $S_2$ with $v(S_1 \bigcup S_2 = N) = 0$. In this case,
no imputation can be found to lie in the core, since the value of the grand coalition is $v(N)=0$. Thus, in such a case, $S_1$ and $S_2$ will
have a better performance in a non-cooperative mode and the core of the transmitter cooperation ($N$,$v$) game is empty.
\end{proof}

As a result of the non-superadditivity of the game and the emptiness of the core, the grand coalition does
\emph{not} form among cooperating transmitters. Instead, independent disjoint coalitions will form in the network.
Therefore, we seek a novel algorithm for coalition formation that accounts for the properties
of the transmitter cooperation game with cost.

An interesting approach for coalition formation
through simple merge and split operations is given by \cite{KA00}. We define a \emph{collection} of coalitions $S$ in the grand coalition $N$ as the family $S = \{S_{1},\ldots,S_{l}\}$ of mutually disjoint coalitions $S_{i}$ of $N$. In other words, a collection is any arbitrary group of disjoint coalitions $S_i$ of $N$ not necessarily spanning all players of $N$. In addition, a collection $S$ of coalitions encompassing all the players of $N$, that is $\bigcup_{j=1}^{l} S_j = N$ is called a \emph{partition} of $N$. Moreover, the merge and split rules defined in \cite{KA00} are simple operations that allow to modify a partition $T$ of $N$ as follows
\begin{itemize}
\item \textbf{Merge Rule:} Merge any set of coalitions
$\{S_{1},\ldots,S_{k}\}$ where $\sum_{j=1}^{k} v(S_{j}) <
v(\bigcup_{j=1}^{k}S_{j})$; thus $\{S_{1},\ldots,S_{k}\}
\rightarrow \bigcup_{j=1}^{k}S_{j}$.

\item \textbf{Split Rule:} Split any set of coalitions
$\bigcup_{j=1}^{k}S_{j}$ where $\sum_{j=1}^{k} v(S_{j}) >
v(\bigcup_{j=1}^{k}S_{j})$; thus $\bigcup_{j=1}^{k}S_{j}
\rightarrow \{S_{1},\ldots,S_{k}\}$.
\end{itemize}
As a result, a group of coalitions (or users) decides to merge if it is able to improve its total utility through the merge; while
a coalition splits into smaller coalitions if it is able to improve the total utility.  Moreover, it is proven in \cite{KA00} that any iteration of  successive arbitrary merge and split operations \emph{terminates}.

A coalition formation algorithm based on merge and split can be formulated for wireless networks. For instance, for the transmitter cooperation game, each stage of our coalition formation algorithm will run in two consecutive phases shown in Table~\ref{tab:tableAlgo}: adaptive coalition formation, and then transmission. During the coalition formation phase, the users form coalitions through an iteration of arbitrary merge and split rules repeated until termination. Following the self organization of the network into coalitions, TDMA transmission takes place with each coalition transmitting in its allotted slots. Subsequently, the transmission phase may occur several times prior to the repetition of the coalition formation phase, notably in low mobility environments where changes in the coalition structure due to mobility are seldom.

\subsection{Stability Notions}\label{sec:stab}
The work done in \cite{KA00} studies the stability of a partition through the concept of defection function.
\begin{table}
\caption{One stage of the proposed merge and split algorithm}
\begin{center}
\begin{tabular}{|l|}
  \hline
  Step 1: Coalition Formation Phase: Arbitrary Merge and Split Rules\\
  \hline
  Step 2: Transmission Phase: One Coalition per Slot\\
  \hline

\end{tabular}\label{tab:tableAlgo}
\end{center}
\end{table}
\begin{definition}
A \emph{defection} function $\mathbb{D}$ is a function which
associates with any arbitrary partition $T=\{T_1,\ldots,T_l\}$ (each $T_i$ is a coalition) of the players set $N$ a family (i.e.~group) of
collections in $N$.
\end{definition}

Two important defection functions can be pinpointed. First, the $\mathbb{D}_{hp}(T)$ function ($\mathbb{D}_{hp}$) which associates with each partition $T$ of $N$
the family of all partitions of $N$ that the players can form through merge and split operations applied to $T$. This function allows any group of players to
leave the partition $T$ of $N$ through \emph{merge and split} operations to create another \emph{partition} in $N$. Second, the $\mathbb{D}_{c}(T)$ function
($\mathbb{D}_{c}$) which associates with each partition $T$ of $N$ the family of all collections in $N$. This function allows any group of players to leave
the partition $T$ of $N$ through \emph{any} operation and create an arbitrary \emph{collection} in $N$. Two forms of stability stem from these definitions: $\mathbb{D}_{hp}$ stability and a stronger strict $\mathbb{D}_{c}$ stability. In fact, a partition $T$ is $\mathbb{D}_{hp}$-stable, if no players in $T$ are interested in leaving $T$ through merge and split to form other partitions in $N$; while a partition $T$ is strictly $\mathbb{D}_{c}$-stable, if no players in $T$ are interested in leaving $T$ to form other collections in $N$ (not necessarily by merge and split).  
\begin{theorem}
Every partition resulting from our proposed merge and split algorithm is $\mathbb{D}_{hp}$-stable.
\end{theorem}
\begin{proof}
A network partition $T$ resulting from the proposed merge and split algorithm can no longer be subject to any additional merge or split operations as successive iteration of these operations terminate \cite{KA00}. Therefore, the users in the final network partition $T$ cannot leave this partition through merge and split and the partition $T$ is immediately $\mathbb{D}_{hp}$-stable.
\end{proof}
Nevertheless, a stronger form of stability can be sought using strict $\mathbb{D}_{c}$-stability. The appeal of a strictly $\mathbb{D}_{c}$ stable partition is two fold \cite{KA00}: 1)~it is the unique outcome of any arbitrary iteration of merge and split operations done on any partition of $N$; 2)~it is a partition that maximizes the social wellfare which is the sum of the utilities of all coalitions in a partition. However, the existence of such a partition is not guaranteed. In fact, the authors in \cite{KA00} showed that a partition $T = \{T_{1},\ldots,T_{l}\}$ of the whole space $N$ is strictly $\mathbb{D}_{c}$-stable only if it can fulfill two necessary and sufficient conditions:
\begin{enumerate}
\item For each $i\in \{1,\ldots,l\}$ and each pair of disjoint \emph{coalitions}  $S_1$ and $S_2$ such that $\{S_1 \cup S_2\} \subset T_i$ we have
$v(S_1 \bigcup S_2) > v(S_1) + v(S_2)$.
\item For the partition $T=\{T_1,\ldots,T_l\}$ a coalition $G \subset N$ formed of players belonging to different $T_i \in T$ is $T$-incompatible, that is for no
$i \in \{1,\ldots,l\}$ we have $G\subset T_i$. Strict $\mathbb{D}_{c}$-stability requires that for  all T-incompatible coalitions $G$, $\sum_{i=1}^{l}v(T_{i} \cap G) > v(G)$.
\end{enumerate}
Therefore, in the case where a partition $T$ of $N$ satisfying the above two conditions exists; the proposed algorithm converges to this optimal strictly $\mathbb{D}_{c}$-stable partition since it constitutes a unique outcome of any arbitrary iteration of merge and split. However, if no such partition exists, the proposed algorithm yields a final network partition that is $\mathbb{D}_{hp}$-stable.
In the transmitter cooperation game, the first condition of $\mathbb{D}_{c}$-stability depends on the users location in the network due to cost of cooperation. In fact, it is well known \cite{TE00} that, in an ideal case with no cost for cooperation, as the number of transmit antennas are increased for a fixed power constraint, the overall system's diversity increases. In fact, consider a partition $T=\{T_1,\ldots,T_l\}$ of $N$, and any two disjoint coalitions $S_1$ and $S_2$ such that
$\{S_1 \cup S_2\} \subset T_i$. Assuming that no cost for cooperation exists, the capacity of the coalition $S_1 \bigcup S_2$, denoted $C_{S_1\bigcup S_2}$, is larger than the capacities $C_{S_1}$ and $C_{S_2}$ of the coalitions $S_1$ and $S_2$ acting non-cooperatively (due to the larger number of antennas in $S_1\bigcup S_2$); thus $|S_1 \bigcup S_2|\cdot C_{S_1\bigcup S_2} >|S_1 \bigcup S_2|\cdot \max{(C_{S_1},C_{S_2})}$ with $|S_1 \bigcup S_2|=|S_1|+|S_2|$. As a result $C_{S_1 \bigcup S_2}$ satisfies
\begin{equation}\label{eq:superconv}
\left|S_1 \bigcup S_2\right|\cdot C_{S_1\bigcup S_2} > |S_1|\cdot C_{S_1} + |S_2|\cdot C_{S_2}.
 \end{equation}
 In fact, (\ref{eq:superconv}) yields $v(S_1 \bigcup S_2) > v(S_1) + v(S_2)$ which is the necessary condition to verify the first $\mathbb{D}_{c}$-stability condition. However, due to the cost given by (\ref{eq:cost}) $C_{S_1\bigcup S_2}, C_{S_1} \textrm{ and } C_{S_2}$ can have different power constraints due to the power cost, i.e.~users location, and this condition is not always verified. Therefore, in practical networks, guaranteeing the first condition for existence of a strictly $\mathbb{D}_{c}$-stable partition is random due to the random location of the users. Furthermore, for a partition $T=\{T_1,\ldots,T_l\}$, the second condition of $\mathbb{D}_{c}$-stability is also dependent on the distance between the users in different coalitions $T_i \in T$. In fact, as previously defined, for a partition $T$ a T-incompatible coalition $G$ is a coalition formed out of users belonging to different $T_i \in T$. In order to always guarantee that $\sum_{i=1}^{l}v(T_{i} \cap G) > v(G)$ it suffice to have $v(G)=0$ for all T-incompatible coalitions $G$. In a network partition $T$ where the players belonging to different coalitions $T_i \in T$ are separated by large distances, any T-incompatible coalition $G$ will have $v(G) =0$ based on (\ref{eq:utility}) and, thus, satisfying the second $\mathbb{D}_c$-stability condition.

Finally, the proposed algorithm can be implemented in a
distributed way. Since the user can detect the strength of the
other users' uplink signals, the nearby users can be discovered.
By using a signalling channel, the distributed users can exchange some channel information and then perform the
merge and split algorithm. The signalling for this handshaking can be minimal.

\subsection{Fairness Criteria for Distributions within Coalition}\label{sec:div}
In this section, we present possible fairness criteria for dividing the coalition worth
among its members.
\subsubsection{Equal Share Fairness}
The most simple division method is to divide the \emph{extra} equally
among users. In other words, the utility of user $i$ among the
coalition $S$ is
\begin{equation}
z_i=\frac {1}{|S|} \left( v(S)-\sum_{j\in S} v(\{j\})\right)+v(\{i\}).
\end{equation}

\subsubsection{Proportional Fairness}

The equal share fairness is a very simple and strict fairness
criterion. However, in practice, the user experiencing a good channel
might not be willing to cooperate with a user under
bad channel conditions, if the extra is divided equally. To account
for the channel differences, we use another fairness criterion named
proportional fairness, in which the extra benefit is divided in weights
according to the users' non-cooperative utilities. In other words,
\begin{equation}
z_i=w_i \left( v(S)-\sum_{j\in S} v(\{j\})\right)+v(\{i\}),
\end{equation}
where $\sum_{i\in S} w_i=1$ and within the coalition
\begin{equation}\label{eq:int}
\frac {w_i}{w_j} = \frac {v(\{i\})}{v(\{j\})},
\end{equation}

\begin{figure}[!t]
\begin{center}
\includegraphics[angle=0,width=\columnwidth]{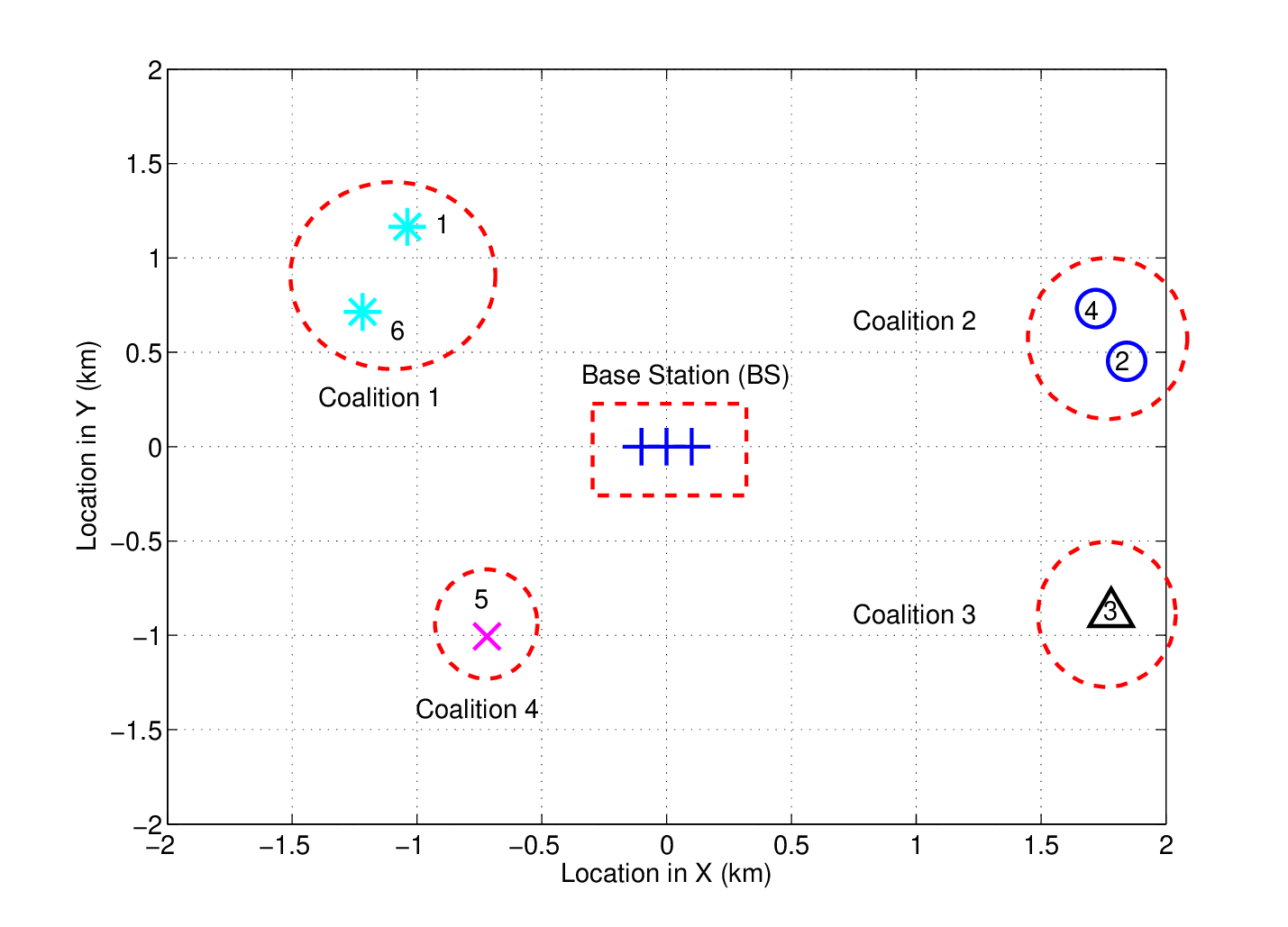}
\end{center}\vspace{-0.5cm}
\caption {A snapshot example of coalition formation.} \label{f:
snapshot}\vspace{3mm}
\end{figure}
\section{Simulation Results} \label{sec:sim}
For simulations, the base station (BS) is placed at the origin with $M_r = 3$ antennas, and random users are located within a square of $2$~km $\times$ $2$~km around the BS. The power constraint per slot is $\tilde{P}=0.01$~W, the SNR for information
exchange is $\nu_0=10$~dB and the noise level is $-90$~dBm. The propagation loss is $\alpha=3$ and $\kappa =1$.

In Fig.~\ref{f: snapshot}, we show a snapshot of a network with
$M_t = 6$ users. Using the proposed merge and split protocol, clusters of users are formed
for distributed closely located users. 
Moreover, the coalition structure in Fig.~\ref{f: snapshot} is strictly $\mathbb{D}_{c}$-stable, thus, it is the unique outcome of any iteration of
merge and split. The strict $\mathbb{D}_{c}$ stability of this structure is immediate since a partition verifying the two conditions of Section~\ref{sec:stab} exists. For the first condition, strict superadditivity within the coalitions is immediately verified by merge rule due to having two users per formed coalition. For the second condition, any T-incompatible coalition will have a utility of $0$ since the users in the different formed coalitions are too far to cooperate. For example,
consider the T-incompatible coalition $\{2,3\}$, the distance between users 3 and 2 is $1.33$~km yielding per (\ref{eq:scost}) and (\ref{eq:cost})
a power cost $\hat{P}_{\{2,3\}} = 0.052\textrm{~W} >\tilde{P}=0.01\textrm{~W}$ thus, by (\ref{eq:utility}) $v(\{2,3\})=0$. This result is easily verifiable for all T-incompatible coalitions.

In Fig.~\ref{f: mobility}, we show how the algorithm handles
mobility. The network setup of Fig.~\ref{f: snapshot} is used and User $6$ moving from the left to right for $2.8$~km. When User
$6$ moves to the right first, it becomes closer to the BS and its utility increases and so does the utility of User $1$. However, when the distance between Users $1$ and $6$ increases, the cost increases and both users' payoffs drop. As long as the distance covered by User $6$ is less than $0.6$~Km, the coalition of Users $1$ and $6$ can still bring mutual benefits to both users. After that, splitting occurs and User $1$ and User $6$ transmit independently. When User $6$ move about $1.2$~Km, it begins to distance itself from the BS and its utility begins to decrease. When User $6$ moves about $2.5$~km, it will be beneficial to users $2$, $4$ and $6$ to form a 3-user coalition through the merge rule since $v(\{2,4,6\}) = 10.8883 > v(\{2,4\}) + v(\{6\}) = 6.5145 +  3.1811  = 9.6956$. As User $6$ moves further away from the BS, User $2$ and User $4$'s utilities are improved within coalition $\{2,4,6\}$, while User $6$'s utility decreases slower than
prior to merge.

\begin{figure}
\begin{center}
\includegraphics[angle=0,width=\columnwidth]{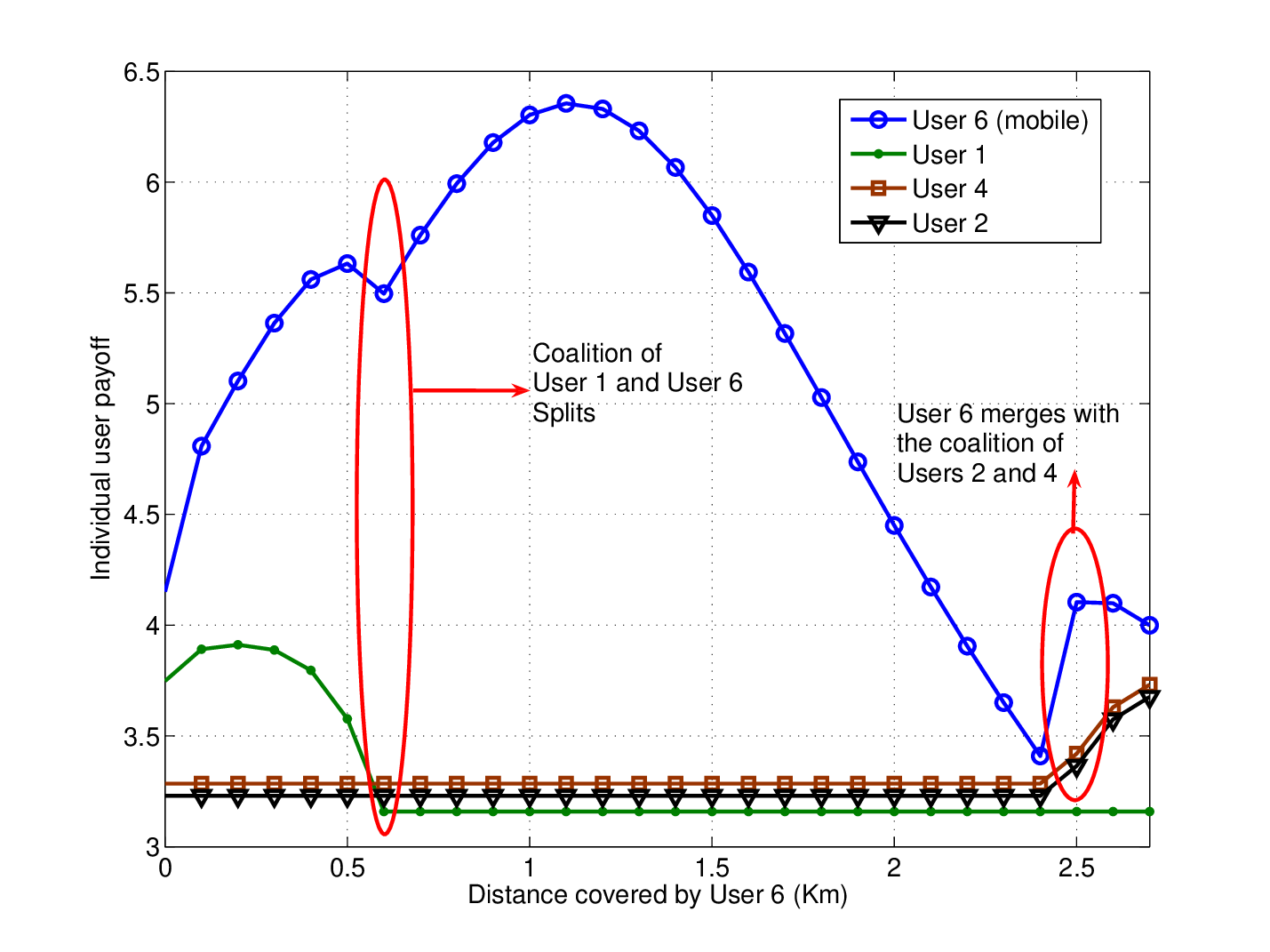}
\end{center}\vspace{-0.5cm}
\caption {Coalition merging/splitting due to mobility.} \label{f:
mobility}
\end{figure}

\begin{table}
\caption{Payoff division according to different fairness}
\begin{center}
\begin{tabular}{|l|c|c|c|}
  \hline
  &User 2&User 4&User 6\\
  \hline
  Equal division& 3.7310 & 3.6761  & 3.9993 \\
  \hline
  Proportional fair & 3.6155 & 3.6968  &4.0940\\

  \hline

\end{tabular}\label{t:pay}
\end{center}
\end{table}

Table~\ref{t:pay} shows the payoff division among
coalition users when the mobile User $6$ moves $2.7$~km in Fig.~\ref{f:
mobility}. In this case, we have $v(\{2\})=2.4422$, $v(\{4\})=2.4971$, $v(\{6\})=2.7654$ and $v(\{2,4,6\})=11.4063$.
Compared with the equal division, proportional fairness gives User $6$ the highest share of the extra benefit and User $2$ the lowest share because User $6$ has a higher non-cooperative utility than User $2$. Thus, Table~\ref{t:pay} shows how different fairness criteria can yield different apportioning of the extra cooperation benefits.

In Fig.~\ref{f: no_user}, we show the average individual user payoff
improvement as a function of the number of users in the networks.
Here we run the simulation for $10000$ different random locations.
For cooperation with coalitions, the average individual utility increases with
the number of users while for the non-cooperative approach an almost constant performance is noted. Cooperation
presents a clear performance advantage reaching up to 26.4\% improvement of the average user payoff at $M_t = 50$ as shown in Fig.~\ref{f: no_user}.
\section{Conclusions} \label{sec:conc}
In this paper, we construct a novel game theoretical algorithm
suitable for modeling distributed cooperation with cost among single antenna
users. Unlike existing literature which sought algorithms to form the grand coalition of transmitters; we inspected the
possibility of forming disjoint independent coalitions which can be characterized by novel stability notions from coalitional
game theory. We proposed a simple and distributed merge and split algorithm for forming coalitions and
benefiting from spatial gains. Various properties of the algorithm are exposed and proved. Simulation results show how the derived algorithm allows the network to self-organize while improving the average user payoff by 26.4\% and efficiently handling the distributed users' mobility.

\begin{figure}
\begin{center}
\includegraphics[angle=0,width=\columnwidth]{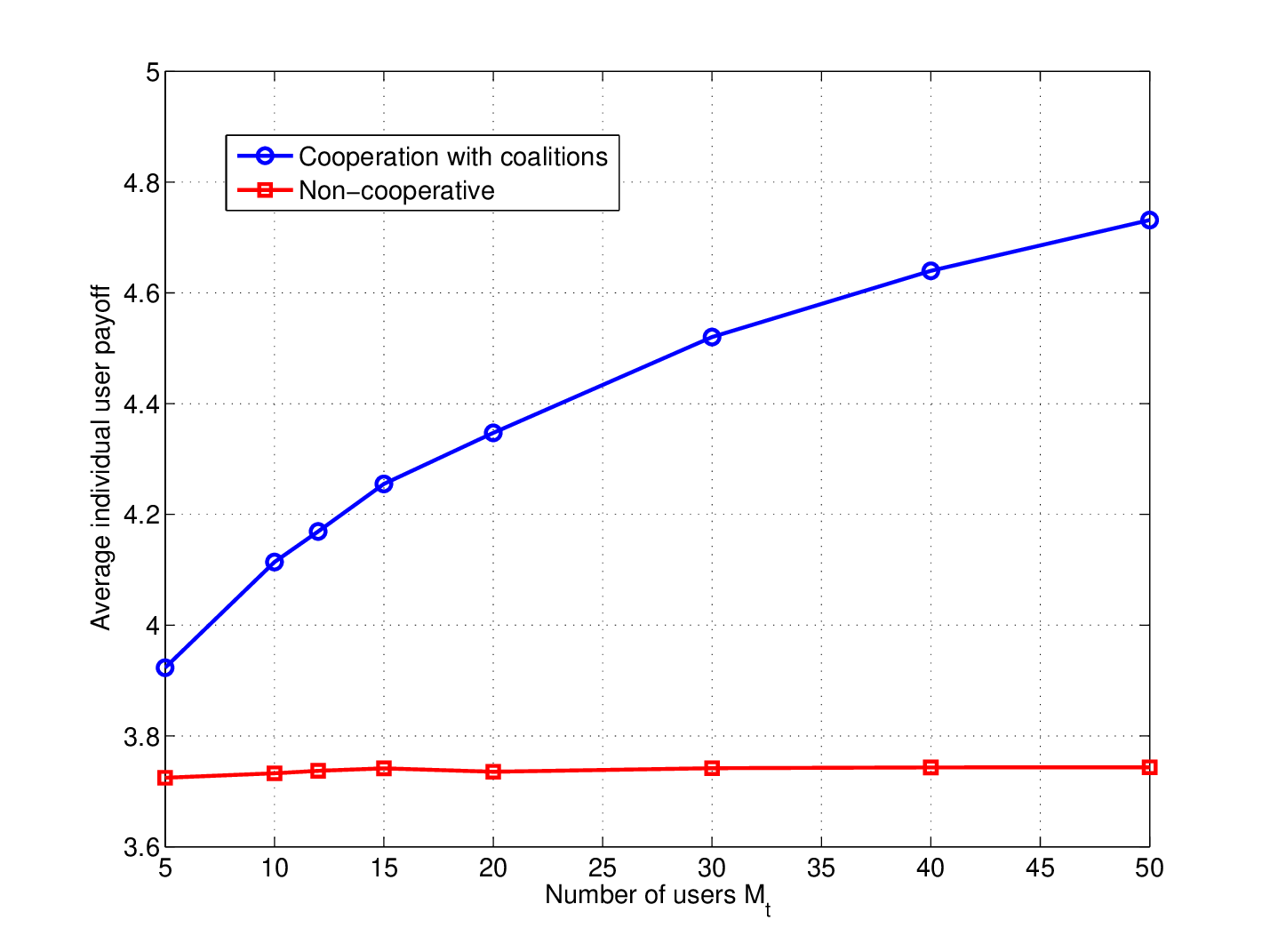}
\end{center}\vspace{-0.5cm}
\caption {Performance improvement.} \label{f: no_user}
\end{figure}
\bibliographystyle{IEEEtran}
\bibliography{references}
\end{document}